\documentclass{eptcs}
\providecommand{\event}{EPTCS} 
\usepackage{breakurl}             

\usepackage{amsmath}
\usepackage{amsthm}
\usepackage{bm}

\usepackage{fancyhdr}
\usepackage{proof}
\usepackage{my-macro}

\title{Uniform Proofs of Normalisation and\\
Approximation for Intersection Types}
\author{Kentaro Kikuchi
\institute{RIEC, Tohoku University\\
Katahira 2-1-1, Aoba-ku, Sendai 980-8577, Japan}
\email{kentaro@nue.riec.tohoku.ac.jp}
}
\def\titlerunning{Uniform Proofs of Normalisation and Approximation for Intersection Types}
\def\authorrunning{K. Kikuchi}
\begin{document}
\maketitle

\begin{abstract}
  We present intersection type systems in the style of sequent
  calculus, modifying the systems that Valentini introduced to prove
  normalisation properties without using the reducibility method.  Our
  systems are more natural than Valentini's ones and equivalent to the
  usual natural deduction style systems.  We prove the
  characterisation theorems of strong and weak normalisation through
  the proposed systems, and, moreover, the approximation theorem
  by means of direct inductive arguments.  This provides in a uniform
  way proofs of the normalisation and approximation theorems via type
  systems in sequent calculus style.
\end{abstract}

\section{Introduction}

A traditional way of proving strong normalisation for typed
$\lambda$-terms is the reducibility method \cite{Tait1967}, 
which uses set-theoretic comprehension.  Other methods without using
reducibility have also been studied in the literature (see,
e.g.\ Section 5 of \cite{vRSSX1999} for a review of those methods).
Some of them use an inductive characterisation of strongly normalising
$\lambda$-terms given by van Raamsdonk and Severi \cite{vRS1995}.  In
\cite{Valentini2001}, Valentini introduced, instead of using the
inductive characterisation, an intersection type system that is closed
under the rules of the original system, and proved strong
normalisation by a simple induction on the typing derivation.

In this paper we develop Valentini's approach further providing an
improvement on his system and its extensions with an axiom for the
type constant $\omega$.  These systems are in the style of sequent
calculus and equivalent to the original intersection type systems in
natural deduction style.  Using the new systems, we prove the
characterisation theorems of strong and weak normalisation, which are
well-known properties of intersection type systems
\cite{Pottinger1980,CDV1981}.

Another important point in our approach is that we design new systems
that derive the same sequents as the original natural deduction style
systems do, so that we can prove various other properties than
normalisation by simple inductions on the typing derivation
(cf.\ \cite{Kikuchi2009}).
In the present paper we illustrate that by showing
the approximation theorem for the type
system with $\omega$, which is usually proved using reducibility
predicates over a typing context and a type
(see, e.g.\ \cite{MR2003d:03016,Barendregt2013}).

The difference between the systems in \cite{Valentini2001} and ours is
the following.  First, some rules of the systems in
\cite{Valentini2001} have restrictions on types to be type variables. 
Also, the rule for abstraction takes a form that implies the
$\eta$-rule.  On the other hand, our systems do not have the
restrictions on types, and our rule for abstraction is the usual one.
In this natural setting, we show that our system is closed under the
rules of the original natural deduction style system.
This part of the proof of strong
normalisation is much shorter than that in \cite{Valentini2001}.
Secondly, the system characterising weakly normalising $\lambda$-terms
in \cite{Valentini2001} does not have the type constant $\omega$, and
is not related to the original natural deduction style system.  In
this paper, we introduce new systems with an axiom for the type
constant $\omega$, and prove weak normalisation of $\lambda$-terms
that are typable with $\omega$-free types in the original system.  The
closure under the rules of the original system is shown by almost the
same argument as that in the case of the system without $\omega$.



In \cite{Valentini2001}, only normalisation properties are discussed,
and other properties than normalisation are not proved using the
sequent calculus style systems.  Some other papers
\cite{vRS1995,Matthes2000,David2001,Abel2007} have studied strong
normalisation for terms typable with intersection types without using
reducibility.  Each of them uses an inductive characterisation of
strongly normalising terms, but any other properties than
normalisation have not been treated.  So the present paper seems to be
the first to apply a proof method for normalisation without
reducibility to other properties of intersection type systems.

There is also an attempt in \cite{vanbakel04strongly} to give uniform
proofs of the characterisation theorems of normalisation and the
approximation theorem.  The method is through strong normalisation for
reduction on typing derivations.  However, it uses reducibility
predicates to prove the strong normalisation, and the proof seems more
complicated than ours.

The organisation of the paper is as follows.  In Section
\ref{section:system} we introduce two kinds of intersection type
systems.  In Section \ref{section:sn} we prove the characterisation
theorem of strong normalisation through the new type system.  In
Section \ref{section:wn} we introduce type systems with $\omega$,
and prove the characterisation theorem of weak normalisation.  In
Section \ref{section:app} we prove
the approximation theorem using one of the new systems with $\omega$.




\section{Intersection type systems}\label{section:system}

In this section we introduce two intersection type systems: one is in
the ordinary natural deduction style and the other in sequent calculus
style.  They prove to be equivalent, and both characterise strongly
normalising $\lambda$-terms.  

First we introduce some basic notions on the $\lambda$-calculus
\cite{Barendregt1984}.  The set $\Lambda$ of $\lambda$-terms
is defined by the grammar: $M::=x\mid MM\mid\lam{x}{M}$ where $x$ ranges
over a denumerable set of variables.  We use letters $x,y,z,\dots$ for
variables and $M,N,P,\dots$ for $\l$-terms.  The notions of free and
bound variables are defined as usual.  The set of free variables occurring
in a $\l$-term $M$ is denoted by $\FV{M}$.  
We identify $\alpha$-convertible $\l$-terms, and use $\equiv$ to denote
syntactic equality modulo $\alpha$-conversion.
${\underline{\ }}\subst{\underline{\ }}{\underline{\ }}$ is used for
usual capture-free substitution.

The $\beta$-rule is stated as $(\lam{x}{M})N\red M\subst{x}{N}$, and
$\beta$-reduction is the contextual closure of the $\beta$-rule.  We
use $\bred$ for one-step reduction, and $\bredstar$ for its reflexive
transitive closure.
A $\l$-term $M$ is said to be \emph{strongly}
(\emph{weakly}) \emph{normalising} if all (some, respectively)
$\beta$-reduction sequences starting from $M$ terminate.  The set of
strongly (weakly) normalising $\l$-terms is denoted by $\SN{\beta}$
($\WN{\beta}$, respectively).

\begin{figure}[b]
\[\begin{array}{|c|}
\upline\\[-12pt]
\infer[(\mathsf{Ax})]{\GA,x:\fA\vd x:\fA}{}
\hspace{3ex}
\infer[(\to\mathsf{I})]{\GA\vd\lam{x}{M}:\fA\to\fB}{
 \GA,x:\fA\vd M:\fB}
\hspace{3ex}
\infer[(\cap\,\mathsf{I})]{\GA\vd M:\fA\cap\fB}{
 \GA\vd M:\fA & \GA\vd M:\fB}\\
\multicolumn{1}{|l|}{\hspace*{22ex}\text{where $x\notin\GA$}}
\\\\
\infer[(\to\mathsf{E})]{\GA\vd MN:\fB}{
 \GA\vd M:\fA\to\fB & \GA\vd N:\fA}
\hspace{3ex}
\infer[(\cap\,\mathsf{E})]{\GA\vd M:\fA}{
 \GA\vd M:\fA\cap\fB}
\hspace{3ex}
\infer[(\cap\,\mathsf{E})]{\GA\vd M:\fB}{
 \GA\vd M:\fA\cap\fB}\\[-12pt]
\downline
\end{array}
\]
\caption{Natural deduction style system $\LA_{\cap}$}\label{table:lambda-inters}
\end{figure}

The set of types is defined by the grammar:
$\fA::=\varphi\mid\fA\to\fA\mid\fA\cap\fA$ where $\varphi$ ranges over
a denumerable set of type variables.  We use letters
$\fA,\fB,\fC,\dots$ for arbitrary types.  The type assignment systems
$\LA_{\cap}$ and $\LA_{\cap}^{\mathit{s}}$ are defined by the rules in
Figures~\ref{table:lambda-inters} and \ref{table:lambda-inters-s},
respectively.  A \emph{typing context} is defined as a finite set of
pairs $\{x_1:\fA_1,\dots,x_n:\fA_n\}$ where the variables are pairwise
distinct in the system $\LA_{\cap}$ while they may be the same in the
system $\LA_{\cap}^{\mathit{s}}$.  A variable with different types is
intended to have the type of intersection of all of them.  The typing
context $\GA,x:\fA$ denotes the union $\GA\cup\{x:\fA\}$, and
$x\notin\GA$ means that $x$ does not appear in $\GA$, i.e., 
for no type $\fA$, $x:\fA\in\GA$.
Note that $x:\fA\in\GA$ is possible in the typing context $\GA,x:\fA$.
In particular, the premisses of the rule $(\mathsf{L}\to)$ may have
$x:\fA_1\to\fA_2$ in $\GA$.
In that case, $x:\fA_1\to\fA_2$ is introduced by the rule
$(\mathsf{L}\to)$ with implicit contraction.


\begin{figure}[t]
\[\begin{array}{|c|}
\upline\\[-12pt]
\infer[(\mathsf{Ax})]{\GA,x:\fA\vds x:\fA}{}
\qquad
\infer[(\mathsf{Beta})^{\mathit{s}}]{\GA\vds(\lam{x}{M})NN_1\dots N_n:\fA}{
 \GA\vds M\subst{x}{N}N_1\dots N_n:\fA & \GA\vds N:\fB}
\\\\
\infer[(\mathsf{L}\to)]{\GA,x:\fA_1\to\fA_2\vds xNN_1\dots N_n:\fB}{
 \GA\vds N:\fA_1 & \GA,y:\fA_2\vds yN_1\dots N_n:\fB}
\qquad
\infer[(\mathsf{R}\to)]{\GA\vds\lam{x}{M}:\fA\to\fB}{
 \GA,x:\fA\vds M:\fB}
\\
\multicolumn{1}{|l|}{\text{where $y\notin\FV{N_1}\cup\dots\cup\FV{N_n}$
 and $y\notin\GA$
\hspace{5ex}where $x\notin\GA$}}\\\\
\infer[(\mathsf{L}\,\cap)]{\GA,x:\fA_1\cap\fA_2\vds xN_1\dots N_n:\fB}{
 \GA,x:\fA_1,x:\fA_2\vds xN_1\dots N_n:\fB}
\qquad
\infer[(\mathsf{R}\,\cap)]{\GA\vds M:\fA\cap\fB}{
 \GA\vds M:\fA & \GA\vds M:\fB}\\[-12pt]
\downline
\end{array}
\]
\caption{Sequent calculus style system $\LA_{\cap}^{\mathit{s}}$}\label{table:lambda-inters-s}
\end{figure}

The system in \cite{Valentini2001} has the restriction in
$\LA_{\cap}^{\mathit{s}}$ that the type $\fA$ in the rules
$(\mathsf{Ax})$ and $(\mathsf{Beta})^{\mathit{s}}$ and the type $\fB$
in the rules $(\mathsf{L}\to)$ and $(\mathsf{L}\,\cap)$ must be type
variables.  Also, the rule $(\mathsf{R}\to)$ takes the following form:
$$
\infer{\GA\vds M:\fA\to\fB}{
 \GA,x:\fA\vds Mx:\fB}
$$
where $x\notin\GA$ and $x\notin\FV{M}$, so that the system includes the
$\eta$-rule and is not equivalent to the system $\LA_{\cap}$.  (For
example, $\vds\lam{x}{x}:(\fA\to\fB)\to((\fC\cap\fA)\to\fB)$ is
derivable in the system of \cite{Valentini2001}, but
$\vd\lam{x}{x}:(\fA\to\fB)\to((\fC\cap\fA)\to\fB)$ is not derivable in
$\LA_{\cap}$.)

\begin{example}
 Self-application can now be typed naturally in $\LA_{\cap}^{\mathit{s}}$, as
 follows (cf.\ \cite[pp.~478--479]{Valentini2001}).
 $$
 \infer[(\mathsf{R}\to)]{\vds\lam{x}{xx}:(\fA\cap(\fA\to\fB))\to\fB}{
  \infer[(\mathsf{L}\,\cap)]{x:\fA\cap(\fA\to\fB)\vds xx:\fB}{
   \infer[(\mathsf{L}\to)]{x:\fA,x:\fA\to\fB\vds xx:\fB}{
    x:\fA\vds x:\fA & x:\fA,y:\fB\vds y:\fB}}}
 $$
\end{example}

The $(\mathsf{Beta})^{\mathit{s}}$-free part of
the system $\LA_{\cap}^{\mathit{s}}$ types
exactly the terms in $\beta$-normal form,
and any $\beta$-redex in a typed term must be constructed
through the rule $(\mathsf{Beta})^{\mathit{s}}$.
So it is immediately seen that
the terms that are not head-normalising
(\eg $(\lam{x}{xx})(\lam{x}{xx})$)
can not be typed in the system $\LA_{\cap}^{\mathit{s}}$.

\begin{proposition}\label{proposition:GA_cap}
 $\GA,x:\fA_1\cap\fA_2\vds M:\fB$ if and only if
 $\GA,x:\fA_1,x:\fA_2\vds M:\fB$.
\end{proposition}
\begin{proof}
 By induction on the derivations.  \qed
\end{proof}

Henceforth we write $\GA_{\cap}$ for the typing context in which each
variable has the type of intersection of all the types that the
variable has in $\GA$.


\section{Characterisation of strongly normalising \texorpdfstring{$\bm{\lambda}$}{lambda}-terms}
\label{section:sn}

If one tries to prove strong normalisation for terms typed in the
system $\LA_{\cap}$ directly by induction on derivations, a difficulty
arises in the case of the rule $(\to\mathsf{E})$.  One way of
overcoming this difficulty is to use reducibility predicates
\cite{Tait1967}.  Here we use the sequent calculus style system
$\LA_{\cap}^{\mathit{s}}$ instead.  For the system
$\LA_{\cap}^{\mathit{s}}$, we can prove strong normalisation for typed
terms directly by induction on derivations.

\begin{theorem}\label{theorem:sn-s}
 If $\GA\vds M:\fA$ then $M\in\SN{\beta}$.
\end{theorem}
\begin{proof}
 By induction on the derivation of $\GA\vds M:\fB$ in $\LA_{\cap}^{\mathit{s}}$.
 The only problematic case is where the last rule applied is
 $(\mathsf{Beta})^{\mathit{s}}$.  In that case, by the induction
 hypothesis, we have $M\subst{x}{N}N_1\dots N_n\in\SN{\beta}$ and
 $N\in\SN{\beta}$.  From the former we have
 $M,N_1,\dots,N_n\in\SN{\beta}$.  Then any infinite reduction
 sequence starting from $(\lam{x}{M})NN_1\dots N_n$ must have the form
 $$
 \begin{array}{lll}
  (\lam{x}{M})NN_1\dots N_n
   & \bredstar & (\lam{x}{M'})N'N_1'\dots N_n' \\[0.5ex]
   & \bred & M'\subst{x}{N'}N_1'\dots N_n' \\[0.5ex]
   & \bred & \;\dots
 \end{array}
 $$
 where $M\bredstar M'$, $N\bredstar N'$ and $N_i\bredstar
 N_i'$ for $i\in\{1,\dots,n\}$.  But then there is an infinite
 reduction sequence
 $$
 \begin{array}{lll}
  M\subst{x}{N}N_1\dots N_n & \bredstar &
   M'\subst{x}{N'}N_1'\dots N_n' \\[0.5ex]
  & \bred & \;\dots
 \end{array}
 $$
 contradicting the hypothesis.  Hence $(\lam{x}{M})NN_1\dots
 N_n\in\SN{\beta}$.  \qed
\end{proof}

To complete a proof of strong normalisation for terms typed in the system
$\LA_{\cap}$, what remains to be shown is that if $M$ is
typable in $\LA_{\cap}$ then it is typable in $\LA_{\cap}^{\mathit{s}}$.
This is proved using several lemmas below.
First we 
show that $\LA_{\cap}^{\mathit{s}}$ is closed under
the weakening rule.

\begin{lemma}\label{lemma:weakening-s}
 If $\GA\vds M:\fB$ then $\GA,x:\fA\vds M:\fB$. 
\end{lemma}
\begin{proof}
 By induction on the derivation of $\GA\vds M:\fB$.  \qed
\end{proof}

The next two lemmas are the essential difference from the proof of
\cite{Valentini2001}.  These are used in the proof of
Lemma~\ref{lemma:subst-s} below.  
The simply typed counterpart of 
Lemma~\ref{lemma:app-var} is found in the second
proof of strong normalisation for the simply typed $\lambda$-calculus
in \cite{JoachimskiMatthes2003}.

\begin{lemma}\label{lemma:app-var}
 If $\GA\vds M:\fA\to\fB$ and $x\notin\GA$ then $\GA,x:\fA\vds Mx:\fB$.
\end{lemma}
\begin{proof}
 By induction on the derivation of $\GA\vds M:\fA\to\fB$.
 Here we show a few cases.
 \begin{itemize}
  \renewcommand{\labelitemi}{$\bullet$}
  \item \raisebox{-1.5ex}[2ex][3ex]{$
	\infer[(\mathsf{Ax})]{\GA,y:\fA\to\fB\vds y:\fA\to\fB}{}
	$} \\
	In this case we take two axioms $\GA,x:\fA\vds x:\fA$ and
	$\GA,x:\fA,z:\fB\vds z:\fB$, and obtain
	$\GA,x:\fA,y:\fA\to\fB\vds yx:\fB$ by an instance of the
	$(\mathsf{L}\to)$ rule.
  \item \raisebox{-1.5ex}[4ex][3ex]{$
	\infer[(\mathsf{Beta})^{\mathit{s}}]{\GA\vds(\lam{y}{M})NN_1\dots N_n:
	\fA\to\fB}{
	 \GA\vds M\subst{y}{N}N_1\dots N_n:\fA\to\fB &
	 \GA\vds N:\fC}
	$} \\
	By the induction hypothesis, we have $\GA,x:\fA\vds
	M\subst{y}{N}N_1\dots N_nx:\fB$, and by
	Lemma~\ref{lemma:weakening-s}, we have $\GA,x:\fA\vds N:\fC$.
	From these, we obtain
	$\GA,x:\fA\vds(\lam{y}{M})NN_1\dots N_nx:\fB$ by an
	instance of the $(\mathsf{Beta})^{\mathit{s}}$ rule.
  \item \raisebox{-1.5ex}[4ex][3ex]{$
	\infer[(\mathsf{R}\to)]{\GA\vds\lam{y}{M}:\fA\to\fB}{
	 \GA,y:\fA\vds M:\fB}
	$} \\
	where $y\notin\GA$.  From $\GA,y:\fA\vds M:\fB$, we have
        $\GA,x:\fA\vds M\subst{y}{x}:\fB$.  From this and the axiom
        $\GA,x:\fA\vds x:\fA$, we obtain
        $\GA,x:\fA\vds(\lam{y}{M})x:\fB$ by an instance of the
        $(\mathsf{Beta})^{\mathit{s}}$ rule.  \qed
 \end{itemize}
\end{proof}

\begin{lemma}\label{lemma:inters-inv}
 If $\GA\vds M:\fA\cap\fB$ then $\GA\vds M:\fA$ and $\GA\vds M:\fB$.
\end{lemma}
\begin{proof}
 By induction on the derivation of $\GA\vds M:\fA\cap\fB$.  \qed
\end{proof}

Now we are in a position to prove the following important lemma.

\begin{lemma}\label{lemma:subst-s}
 $\LA_{\cap}^{\mathit{s}}$ is closed under substitution, i.e.,
 if $\GA,x:\fA_1,\dots,x:\fA_m\vds P:\fB$ where $x\notin\GA$, $m\ge 0$
 and $\fA_i\neq\fA_j$ for $i\neq j$, and, for any $i\in\{1,\dots,m\}$,
 $\GA\vds N:\fA_i$, then $\GA\vds P\subst{x}{N}:\fB$.
\end{lemma}
\begin{proof}
 The proof is by main induction on the number of `$\to$' and `$\cap$'
 occurring in $\fA_1,\dots,\fA_m$ and subinduction on the length of
 the derivation of $\GA,x:\fA_1,\dots,x:\fA_m\vds P:\fB$.  We proceed
 by case analysis according to the last rule used in the derivation of
 $\GA,x:\fA_1,\dots,x:\fA_m\vds P:\fB$.  Here we consider 
 a few cases.
 \begin{itemize}
  \renewcommand{\labelitemi}{$\bullet$}
  \item Suppose the last rule in the derivation is
	$$
	\infer[(\mathsf{Beta})^{\mathit{s}}]{\GA,\overline{x}:\overline{\fA}\vds
	(\lam{y}{M})QN_1\dots N_n:\fB}{
	 \GA,\overline{x}:\overline{\fA}\vds M\subst{y}{Q}N_1\dots N_n:\fB &
	 \GA,\overline{x}:\overline{\fA}\vds Q:\fC}
	$$
	where $\overline{x}:\overline{\fA}=
        x:\fA_1,\dots,x:\fA_m$.
	By the subinduction hypothesis, we obtain both
	$$
	\GA\vds	M\subst{y}{Q}\subst{x}{N}N_1\subst{x}{N}\dots
	N_n\subst{x}{N}:\fB
	$$
	and
	$$
	\GA\vds Q\subst{x}{N}:\fC
	$$
	Since $y$ is a bound variable, we can assume that
	it does not occur in $N$.  Hence the first judgement is
	$$
	\GA\vds	M\subst{x}{N}\subst{y}{Q\subst{x}{N}}N_1\subst{x}{N}\dots
	N_n\subst{x}{N}:\fB
	$$
	From this and $\GA\vds Q\subst{x}{N}:\fC$, we obtain
	$$
	\GA\vds(\lam{y}{M}\subst{x}{N})Q\subst{x}{N}N_1\subst{x}{N}\dots
	N_n\subst{x}{N}:\fB
	$$
	by an instance of the $(\mathsf{Beta})^{\mathit{s}}$ rule.
  \item Suppose the last rule in the derivation is
	$$
	\infer[(\mathsf{L}\to)]{\GA,\overline{x}:\overline{\fA},
	x:\fC_1\to\fC_2\vds xMN_1\dots N_n:\fB}{
	 \GA,\overline{x}:\overline{\fA}\vds M:\fC_1 &
	 \GA,\overline{x}:\overline{\fA},y:\fC_2\vds yN_1\dots N_n:\fB}
	$$
	where $\{\overline{x}:\overline{\fA},x:\fC_1\to\fC_2\}=
        \{x:\fA_1,\dots,x:\fA_m\}$,
        $y\notin\FV{N_1}\cup\dots\cup\FV{N_n}$ and
        $y\notin\GA,\overline{x}:\overline{\fA}$.  By the subinduction
        hypothesis, we obtain both
	\begin{equation}\label{eq:c}
	 \GA\vds M\subst{x}{N}:\fC_1
	\end{equation}
	and
	\begin{equation}\label{eq:d}
	 \GA,y:\fC_2\vds (yN_1\dots N_n)\subst{x}{N}:\fB
	\end{equation}
	Now consider the assumption $\GA\vds N:\fC_1\to\fC_2$ and a
        fresh variable $z$.  Then by Lemma~\ref{lemma:app-var}, we
        have $\GA,z:\fC_1\vds Nz:\fC_2$.  From this and (\ref{eq:c}),
        we have $\GA\vds NM\subst{x}{N}:\fC_2$ by the main induction
        hypothesis.  Then, again by the main induction
	hypothesis, we obtain
	$$
	\GA\vds	NM\subst{x}{N}N_1\subst{x}{N}\dots N_n\subst{x}{N}:\fB
	$$
	from (\ref{eq:d}) and $\GA\vds NM\subst{x}{N}:\fC_2$.
  \item Suppose the last rule in the derivation is
	$$
	\infer[(\mathsf{L}\,\cap)]{\GA,\overline{x}:\overline{\fA},
	x:\fC_1\cap\fC_2\vds xN_1\dots N_n:\fB}{
	 \GA,\overline{x}:\overline{\fA},x:\fC_1,x:\fC_2\vds
	 xN_1\dots N_n:\fB}
	$$
	where $\{\overline{x}:\overline{\fA},x:\fC_1\cap\fC_2\}=
        \{x:\fA_1,\dots,x:\fA_m\}$.  Then, applying
        Proposition~\ref{proposition:GA_cap} to the conclusion, we
        have $\GA,(\overline{x}:\overline{\fA})',x:\fC_1,x:\fC_2\vds
        xN_1\dots N_n:\fB$ where $(\overline{x}:\overline{\fA})'=
        \overline{x}:\overline{\fA}\setminus\{x:\fC_1\cap\fC_2\}$.
        Now, from the assumption $\GA\vds N:\fC_1\cap\fC_2$, we have
        $\GA\vds N:\fC_1$ and $\GA\vds N:\fC_2$ by
        Lemma~\ref{lemma:inters-inv}.  Hence, by the main induction
        hypothesis, we obtain $\GA\vds NN_1\subst{x}{N}\dots
        N_n\subst{x}{N}:\fB$.  \qed
 \end{itemize}
\end{proof}


Now we can show that the system $\LA_{\cap}^{\mathit{s}}$ is closed
under the $(\to\mathsf{E})$ rule.

\begin{lemma}\label{lemma:app}
 If $\GA\vds M:\fA\to\fB$ and $\GA\vds N:\fA$ then $\GA\vds MN:\fB$.
\end{lemma}
\begin{proof}
 By Lemma~\ref{lemma:app-var}, we have $\GA,x:\fA\vds Mx:\fB$ for any
 fresh variable $x$.  Hence by the previous lemma, we obtain $\GA\vds
 (Mx)\subst{x}{N}\equiv MN:\fB$.  \qed
\end{proof}

Now we can prove the announced theorem.

\begin{theorem}\label{theorem:n-to-s}
 If $\GA\vd M:\fA$ then $\GA\vds M:\fA$.
\end{theorem}
\begin{proof}
 By induction on the derivation of $\GA\vd M:\fA$ in $\LA_{\cap}$, using
 Lemmas~\ref{lemma:inters-inv} and \ref{lemma:app}.  \qed
\end{proof}


The converse of this theorem also holds when typing contexts are
restricted to those of $\LA_{\cap}$.  To prove it, we need some lemmas
on properties of the system $\LA_{\cap}$.

\begin{lemma}\label{lemma:weakening-n}
 If $\GA\vd M:\fB$ and $z\notin\GA$ then $\GA,z:\fA\vd M:\fB$. 
\end{lemma}
\begin{proof}
 By induction on the derivation of $\GA\vd M:\fB$.  \qed
\end{proof}

\begin{lemma}\label{lemma:subst-n}
 $\LA_{\cap}$ is closed under substitution, i.e., if $\GA,x:\fA\vd
 P:\fB$ where $x\notin\GA$ and $\GA\vd N:\fA$ then $\GA\vd
 P\subst{x}{N}:\fB$.
\end{lemma}
\begin{proof}
 By induction on the derivation of $\GA,x:\fA\vd P:\fB$.  \qed
\end{proof}

Next we prove a Generation Lemma.  For its statement we define a
preorder on types.

\begin{definition}\label{def:le}
 The relation $\le$ on types is defined by the following axioms and
 rules:
 \begin{align*}
  1.\ \,& \fA\le\fA & 3.\ \,& \fA\le\fB,\
  \fB\le\fC\,\Rightarrow\,\fA\le\fC \\
  2.\ \,& \fA\cap\fB\le\fA,\ \fA\cap\fB\le\fB & 4.\ \,& \fA\le\fB,\
  \fA\le\fC\,\Rightarrow\,\fA\le\fB\cap\fC
 \end{align*}
\end{definition}

\begin{lemma}\label{lemma:le-n}
 If $\GA\vd M:\fA$ and $\fA\le\fB$ then $\GA\vd M:\fB$.
\end{lemma}
\begin{proof}
 By induction on the definition of $\fA\le\fB$.  \qed
\end{proof}

\begin{lemma}[Generation Lemma]\label{lemma:generation-n}\strut
 \begin{enumerate}
  \item $\GA\vd MN:\fA$ if and only if there exist
	$\fA_1,\dots,\fA_n,\fB_1,\dots,\fB_n$ $(n\ge 1)$ such that
	$\fA_1\cap\dots\cap\fA_n\le\fA$ and, for all $i\in\{1,\dots,n\}$,
	$\GA\vd M:\fB_i\to\fA_i$ and $\GA\vd N:\fB_i$.
  \item $\GA\vd\lam{x}{M}:\fA$ if and only if there exist
	$\fB_1,\dots,\fB_n,\fC_1,\dots,\fC_n$ $(n\ge 1)$ such that
	$(\fB_1\to\fC_1)\cap\dots\cap(\fB_n\to\fC_n)\le\fA$ and, for all
	$i\in\{1,\dots,n\}$, $\GA,x:\fB_i\vd M:\fC_i$.
 \end{enumerate}
\end{lemma}
\begin{proof}
 The implications from right to left are immediate by the typing rules
 and Lemma~\ref{lemma:le-n}.  The converses are shown by induction on
 the derivations.  \qed
\end{proof}

Now we can prove a crucial lemma about type-checking in the system
$\LA_{\cap}$.

\begin{lemma}\label{lemma:inv-subst-n}
 If $\GA\vd M\subst{x}{N}:\fA$ and $\GA\vd N:\fB$ where $x\notin\GA$
 then there exists a type $\fC$ such that $\GA,x:\fC\vd M:\fA$ and
 $\GA\vd N:\fC$.
\end{lemma}
\begin{proof}
 By induction on the structure of $M$, using
 Lemma~\ref{lemma:generation-n}.  \qed
\end{proof}

We are now ready to prove the equivalence between the systems
$\LA_{\cap}^{\mathit{s}}$ and $\LA_{\cap}$.


\begin{theorem}\label{theorem:s-eq-n}
 $\GA\vds M:\fA$ if and only if $\GA_{\cap}\vd M:\fA$.
\end{theorem}
\begin{proof}
 The implication from right to left follows from
 Theorem~\ref{theorem:n-to-s} and
 Proposition~\ref{proposition:GA_cap}.  The converse is shown by
 induction on the derivation of $\GA\vds M:\fA$.  If the last applied
 rule is $(\mathsf{Beta})^{\mathit{s}}$, we use
 Lemmas~\ref{lemma:generation-n} and \ref{lemma:inv-subst-n}.  \qed
\end{proof}


Finally we show that all strongly normalising terms are typable in
$\LA_{\cap}^{\mathit{s}}$.

\begin{theorem}\label{theorem:sn-typable-s}
 If $M\in\SN{\beta}$ then there exist a typing context $\GA$ and a
 type $\fA$ such that $\GA\vds M:\fA$.
\end{theorem}
\begin{proof}
 The proof is by main induction on the maximal length of all
 $\beta$-reduction sequences starting from $M$ and subinduction on the
 structure of $M$.  We analyse the possible cases according to the
 shape of the term $M$.
 \begin{itemize}
  \renewcommand{\labelitemi}{$\bullet$}
  \item $M\equiv x$ for some variable $x$.  
	In this case we just have to take $x:\fA\vds x:\fA$, which
	is an axiom.
  \item $M\equiv xN_1\dots N_n$.  By the subinduction hypothesis, for
    any $i\in\{1,\dots,n\}$, there exist a typing context $\GA_i$ and
    a type $\fA_i$ such that $\GA_i\vds N_i:\fA_i$.  Then consider the
    following derivation (recall that $\LA_{\cap}^{\mathit{s}}$ is
    closed under the weakening rule):
$$
 \infer[(\mathsf{L}\to)]{\cup\GA_i,x:\fA_1\to\dots\to\fA_n\to\fB\vds xN_1\dots N_n:\fB}{
   \cup\GA_i\vds N_1:\fA_1 &
   \infer[(\mathsf{L}\to)]{\cup\GA_i,y_1:\fA_2\to\dots\to\fA_n\to\fB\vds y_1N_2\dots N_n:\fB}{
    \cup\GA_i\vds N_2:\fA_2 &
    \infer*{\cup\GA_i,y_2:\fA_3\to\dots\to\fA_n\to\fB\vds y_2N_3\dots N_n:\fB}{
     \infer[(\mathsf{L}\to)]{\cup\GA_i,y_{n-1}:\fA_n\to\fB\vds y_{n-1}N_n:\fB}{
      \cup\GA_i\vds N_n:\fA_n & \cup\GA_i,y_n:\fB\vds y_n:\fB}}}}
$$
  \item $M\equiv \lam{x}{P}$.  By the subinduction hypothesis, there
	exist a typing context $\GA$ and a type $\fA$ such that
	$\GA,x:\fA_1,\dots,x:\fA_n\vds P:\fA$ where $x\notin\GA$ and $n\ge 0$.
	Then we have $\GA\vds\lam{x}{P}:\fA_1\cap\dots\cap\fA_n\to\fA$
	by the $(\mathsf{L}\,\cap)$ and $(\mathsf{R}\to)$ rules.  (We use
	a weakening rule instead of $(\mathsf{L}\,\cap)$ when $n=0$.)
  \item $M\equiv (\lam{x}{P})NN_1\dots N_n$.  By the main induction
	hypothesis, there exist a typing context $\GA_1$ and a type
	$\fA_1$ such that $\GA_1\vds P\subst{x}{N}N_1\dots N_n:\fA_1$,
	and, by the subinduction hypothesis, there exist a typing context
	$\GA_2$ and a type $\fA_2$ such that $\GA_2\vds N:\fA_2$.  Then,
	by the weakening and $(\mathsf{Beta})^{\mathit{s}}$ rules, we
	obtain $\GA_1,\GA_2\vds(\lam{x}{P})NN_1\dots N_n:\fA_1$.  \qed
 \end{itemize}
\end{proof}

It is interesting to note that in the above proof we do not use the
$(\mathsf{R}\,\cap)$ rule at all, so it is redundant for
characterising the strongly normalising $\l$-terms. The absence of the
$(\mathsf{R}\,\cap)$ rule leads to a restriction on types that is
similar to those investigated in~\cite{vanBakel92}.

The results in this section are summarised as follows.

\begin{corollary}
 For any $\l$-term $M$, the following are equivalent.
 \begin{enumerate}
  \item $M$ is typable in $\LA_{\cap}$.
  \item $M$ is typable in $\LA_{\cap}^{\mathit{s}}$.
  \item $M$ is strongly normalising.
  \item $M$ is typable in $\LA_{\cap}^{\mathit{s}}$ without using the
        $(\mathsf{R}\,\cap)$ rule.
 \end{enumerate}
\end{corollary}
\begin{proof}
 ($1\Rightarrow 2$) This follows from Theorem~\ref{theorem:n-to-s}. \\
 ($2\Rightarrow 3$) This follows from Theorem~\ref{theorem:sn-s}. \\
 ($3\Rightarrow 4$) This follows from the proof of
 Theorem~\ref{theorem:sn-typable-s}. \\
 ($4\Rightarrow 2$) This is trivial. \\
 ($2\Rightarrow 1$) This follows from Theorem~\ref{theorem:s-eq-n}.
 \qed
\end{proof}


\section{Characterisation of weakly normalising \texorpdfstring{$\bm{\lambda}$}{lambda}-terms}
\label{section:wn}

In this section we are concerned with weak normalisation and some type
systems obtained by extending the systems $\LA_{\cap}$ and
$\LA_{\cap}^{\mathit{s}}$.  The main goal of this section is to prove
the characterisation theorem of weak normalisation in a similar way to
that of strong normalisation in the previous section.


The extended systems are listed in Figure~\ref{table:extended-systems}.
First we introduce a new rule $(\mathsf{Beta})^{\ell}$, which is
a general form of the rule considered in \cite{Valentini2001} ($\fA$
is restricted to type variables in \cite{Valentini2001}).  Then the
system $\LA_{\cap}^{\ell}$ is obtained from
$\LA_{\cap}^{\mathit{s}}$ by replacing the
$(\mathsf{Beta})^{\mathit{s}}$ rule by the
$(\mathsf{Beta})^{\ell}$ rule.  The systems $\LA_{\cap\omega}$,
$\LA_{\cap\omega}^{\mathit{s}}$ and $\LA_{\cap\omega}^{\ell}$
are obtained from $\LA_{\cap}$, $\LA_{\cap}^{\mathit{s}}$ and
$\LA_{\cap}^{\ell}$, respectively, by adding the type constant
$\omega$ and the $(\omega)$ rule.  In order to distinguish the
judgements of the systems, we use the symbols $\vdl$, $\vdo$, $\vdso$
and $\vdlo$.

\begin{figure}[b]
\[\begin{array}{|c|}
\upline\\[-12pt]
\infer[(\mathsf{Beta})^{\ell}]{\GA\vd(\lam{x}{M})NN_1\dots N_n:\fA}{
 \GA\vd M\subst{x}{N}N_1\dots N_n:\fA}
\qquad\quad
\infer[(\omega)]{\GA\vd M:\omega}{}
\\\\
\begin{array}{llll}
 &&& \text{ Notation} \\[0.5ex]
 \LA_{\cap}^{\ell} & := & \LA_{\cap}^{\mathit{s}}-(\mathsf{Beta})^{\mathit{s}}+(\mathsf{Beta})^{\ell}
  \hspace*{5ex} & \GA\vdl M:\fA \\[0.5ex]
 \LA_{\cap\omega} & := & \LA_{\cap}+(\omega) & \GA\vdo M:\fA \\[0.5ex]
 \LA_{\cap\omega}^{\mathit{s}} & := & \LA_{\cap}^{\mathit{s}}+(\omega) & \GA\vdso M:\fA \\[0.5ex]
 \LA_{\cap\omega}^{\ell} & := & \LA_{\cap}^{\ell}+(\omega) & \GA\vdlo M:\fA
\end{array}\\[-12pt]
\downline
\end{array}
\]
\caption{Systems extended with $\omega$}\label{table:extended-systems}
\end{figure}

For the system $\LA_{\cap}^{\ell}$, we have the following theorem.

\begin{theorem}\label{theorem:wn-l}
 If $\GA\vdl M:\fA$ then $M\in\WN{\beta}$.
\end{theorem}
\begin{proof}
 By induction on the derivation of $\GA\vdl M:\fB$.  \qed
\end{proof}

For characterisation of weak normalisation in terms of typability in
the extended systems, it is necessary to clarify the relationship
among them.  First we show that the terms typable in the ordinary
natural deduction style system $\LA_{\cap\omega}$ are typable in
$\LA_{\cap\omega}^{\mathit{s}}$, in almost the same way as in the
previous section.

\begin{theorem}\label{theorem:no-to-so}
 If $\GA\vdo M:\fA$ then $\GA\vdso M:\fA$.
\end{theorem}
\begin{proof}
 It is easy to see that Lemmas \ref{lemma:weakening-s} through
 \ref{lemma:app} hold for $\LA_{\cap\omega}^{\mathit{s}}$ instead of
 $\LA_{\cap}^{\mathit{s}}$.  Then the theorem follows by induction on the
 derivation of $\GA\vdo M:\fA$ in $\LA_{\cap\omega}$.  \qed
\end{proof}

Next we relate the systems $\LA_{\cap\omega}^{\mathit{s}}$,
$\LA_{\cap\omega}^{\ell}$ and $\LA_{\cap}^{\ell}$.  This
completes one direction of the characterisation theorem of weak
normalisation.

\begin{lemma}\label{lemma:so-eq-lo}
 $\GA\vdso M:\fA$ if and only if $\GA\vdlo M:\fA$.
\end{lemma}
\begin{proof}
 The implication from left to right is immediate by forgetting the
 right premiss of $(\mathsf{Beta})^{\mathit{s}}$.  For the converse,
 observe that the $(\mathsf{Beta})^{\ell}$ rule is derivable in
 $\LA_{\cap\omega}^{\mathit{s}}$ using the rules
 $(\mathsf{Beta})^{\mathit{s}}$ and $(\omega)$.  \qed
\end{proof}




\begin{lemma}\label{lemma:lo-eq-l}
 Suppose $\fA$ and all types in $\GA$ are $\omega$-free.  Then $\GA\vdlo
 M:\fA$ if and only if $\GA\vdl M:\fA$.
\end{lemma}
\begin{proof}
 The implication from right to left is trivial.  For the converse,
 observe that every type occurring in the derivation of $\GA\vdlo
 M:\fA$ also occurs in $\GA$ or $\fA$.  \qed
\end{proof}

\begin{corollary}
 If $\GA\vdo M:\fA$ where $\fA$ and all types in $\GA$ are
 $\omega$-free, then $M\in\WN{\beta}$.
\end{corollary}
\begin{proof}
 By Theorem \ref{theorem:no-to-so}, Lemmas \ref{lemma:so-eq-lo} and
 \ref{lemma:lo-eq-l}, and Theorem \ref{theorem:wn-l}.  \qed
\end{proof}

Conversely, if a $\l$-term $M$ is weakly normalising, then there exist
a typing context $\GA$ and a type $\fA$, both $\omega$-free, such that
$\GA\vdo M:\fA$.
To prove this, we need the following lemmas on properties of the
system $\LA_{\cap\omega}$.  These are shown in similar ways to the
proofs of Lemmas \ref{lemma:weakening-n} through
\ref{lemma:generation-n}.

\begin{lemma}\label{lemma:weakening-no}
 If $\GA\vdo M:\fB$ and $z\notin\GA$ then $\GA,z:\fA\vdo M:\fB$.
\end{lemma}

\begin{lemma}\label{lemma:subst-no}
 $\LA_{\cap}$ is closed under substitution, i.e., if $\GA,x:\fA\vdo
 P:\fB$ where $x\notin\GA$ and $\GA\vdo N:\fA$ then $\GA\vdo
 P\subst{x}{N}:\fB$.
\end{lemma}

\begin{definition}
 The relation $\le_{\omega}$ on types is defined by the axioms and rules
 in Definition~\ref{def:le} together with the axiom $\fA\le_{\omega}\omega$.
\end{definition}

\begin{lemma}\label{lemma:le-no}
 If $\GA\vdo M:\fA$ and $\fA\le_{\omega}\fB$ then $\GA\vdo M:\fB$.
\end{lemma}

\begin{lemma}[Generation Lemma]\label{lemma:generation-no}
 Let $\fA$ be any type with $\omega\not\le_{\omega}\fA$. Then
 \begin{enumerate}
  \item $\GA\vdo MN:\fA$ if and only if
	there exist $\fA_1,\dots,\fA_n,\fB_1,\dots,\fB_n$ $(n\ge 1)$
	such that $\fA_1\cap\dots\cap\fA_n\le_{\omega}\fA$ and, for all
	$i\in\{1,\dots,n\}$, $\GA\vdo M:\fB_i\to\fA_i$ and $\GA\vdo N:\fB_i$.
  \item $\GA\vdo\lam{x}{M}:\fA$ if and
	only if there exist $\fB_1,\dots,\fB_n,\fC_1,\dots,\fC_n$ $(n\ge
	1)$ such that
	$(\fB_1\to\fC_1)\cap\dots\cap(\fB_n\to\fC_n)\le_{\omega}\fA$
	and, for all $i\in\{1,\dots,n\}$, $\GA,x:\fB_i\vdo M:\fC_i$.
 \end{enumerate}
\end{lemma}

Now we can prove a crucial lemma about type-checking in the system
$\LA_{\cap\omega}$.

\pagebreak

\begin{lemma}\label{lemma:inv-subst-no}
 If $\GA\vdo M\subst{x}{N}:\fA$ where $x\notin\GA$ then there exists a
 type $\fC$ such that $\GA,x:\fC\vdo M:\fA$ and $\GA\vdo N:\fC$.
\end{lemma}
\begin{proof}
 By induction on the structure of $M$, using
 Lemma~\ref{lemma:generation-no}.  If $M\equiv y(\not\equiv x)$ or
 $\omega\le_{\omega}\fA$, then we take $\fC=\omega$.  \qed
\end{proof}

We can now prove that in the system $\LA_{\cap\omega}$, types are
preserved under the inverse of $\beta$-reduction.

\begin{lemma}\label{lemma:sub-exp}
 If $\GA\vdo N:\fA$ and $M\bred N$ then $\GA\vdo M:\fA$.
\end{lemma}
\begin{proof}
 By induction on the structure of $M$, using
 Lemma~\ref{lemma:generation-no}.  If $M$ is the $\beta$-redex then we
 use Lemma~\ref{lemma:inv-subst-no}.  \qed
\end{proof}

Now we can prove the announced theorem.

\begin{theorem}\label{theorem:wn-typable-no}
 If $M\in\WN{\beta}$ then there exist a typing context $\GA$ and a
 type $\fA$ such that $\GA\vdo M:\fA$ and both $\GA$ and $\fA$ are
 $\omega$-free.
\end{theorem}
\begin{proof}
 Let $M'$ be a normal form of $M$.  By
 Theorem~\ref{theorem:sn-typable-s}, every normal form is typable in
 $\LA_{\cap}^{\mathit{s}}$, so there exist a typing context $\GA$ and
 a type $\fA$, both $\omega$-free, such that $\GA\vdo M':\fA$.  Hence,
 by Lemma~\ref{lemma:sub-exp}, we have $\GA\vdo M:\fA$.  \qed
\end{proof}

We can also prove the equivalence of the systems $\LA_{\cap\omega}$,
$\LA_{\cap\omega}^{\mathit{s}}$ and $\LA_{\cap\omega}^{\ell}$.

\begin{theorem}\label{theorem:no-eq-so-eq-lo}
 For any typing context $\GA$, any $\l$-term $M$ and any type $\fA$,
 the following are equivalent. 
 \begin{enumerate}
  \item $\GA_{\cap}\vdo M:\fA$.
  \item $\GA\vdso M:\fA$.
  \item $\GA\vdlo M:\fA$.
 \end{enumerate}
\end{theorem}
\begin{proof}
 ($1\Rightarrow 2$) This follows from Theorem~\ref{theorem:no-to-so} and
 Proposition~\ref{proposition:GA_cap} with $\vdso$ instead of $\vds$. \\
 ($2\Rightarrow 3$) This follows from Lemma~\ref{lemma:so-eq-lo}. \\
 ($3\Rightarrow 1$) This follows by induction on the length of the
 derivation of $\GA\vdlo M:\fA$.  If the last applied rule is
 $(\mathsf{Beta})^{\ell}$, we use Lemmas \ref{lemma:generation-no} and
 \ref{lemma:inv-subst-no}.  \qed
\end{proof}

The results in this section are summarised as follows.

\begin{corollary}
 For any $\l$-term $M$, the following are equivalent.
 \begin{enumerate}
  \item $\GA\vdo M:\fA$ for some typing context $\GA$ and type $\fA$, both $\omega$-free.
  \item $\GA\vdso M:\fA$ for some typing context $\GA$ and type $\fA$, both $\omega$-free.
  \item $\GA\vdlo M:\fA$ for some typing context $\GA$ and type $\fA$, both $\omega$-free.
  \item $\GA\vdl M:\fA$ for some typing context $\GA$ and type $\fA$.
  \item $M$ is weakly normalising.
 \end{enumerate}
\end{corollary}
\begin{proof}
 ($1\Rightarrow 2$) This follows from Theorem~\ref{theorem:no-to-so}. \\
 ($2\Rightarrow 3$) This follows from Lemma~\ref{lemma:so-eq-lo}. \\
 ($3\Rightarrow 4$) This follows from Lemma~\ref{lemma:lo-eq-l}. \\
 ($4\Rightarrow 5$) This follows from Theorem~\ref{theorem:wn-l}. \\
 ($5\Rightarrow 1$) This follows from Theorem~\ref{theorem:wn-typable-no}.
 \qed
\end{proof}


\section{Application to other properties}\label{section:app}

The sequent calculus style systems we introduced in the previous
sections are very useful for proving properties of intersection type
systems.  In this section we illustrate that by giving a simple proof
of the (logical) approximation theorem, a property that is usually
proved using reducibility predicates para\-metrised by typing contexts
(see, e.g.\ \cite{MR2003d:03016,Barendregt2013}).  Proofs of some
other properties 
through the sequent calculus style systems are found
in \cite{Kikuchi2009}, which also makes a comparison between general
conditions for applying the reducibility method and our approach.

For the statement of the approximation theorem, we introduce some
preliminary definitions.  The set of $\lambda\bot$-terms
\cite{Barendregt1984} is obtained by adding the constant $\bot$ to the
formation rules of $\lambda$-terms.
The type systems in the previous section are extended to those for
$\lambda\bot$-terms, where any $\lambda\bot$-term containing $\bot$ 
is typable by the $(\omega)$ rule.


\begin{definition}
 The \emph{approximation mapping} $\alpha$ from $\lambda$-terms to
 $\lambda\bot$-terms is defined inductively by
 \begin{align*}
  &\alpha(\lam{x_1\dots x_n}{xN_1\dots N_m}):=
   \lam{x_1\dots x_n}{x\alpha(N_1)\dots\alpha(N_m)} \\
  &\alpha(\lam{x_1\dots x_n}{(\lam{x}{M})NN_1\dots N_m}):=
   \lam{x_1\dots x_n}{\bot}
 \end{align*}
 where $n,m\ge 0$.
\end{definition}


\begin{lemma}\label{lemma:approx}\strut
 \begin{enumerate}
  \item\label{sub-red-app} If $\GA\vdlo\alpha(M):\fA$ and $M\bredstar N$
    then $\GA\vdlo\alpha(N):\fA$.
  \item\label{inters-app} Let $M\bredstar N$, $M\bredstar N'$,
    $\GA\vdlo\alpha(N):\fA$ and $\GA\vdlo\alpha(N'):\fB$.  Then there
    exists $N''$ such that $M\bredstar N''$ and
    $\GA\vdlo\alpha(N''):\fA\cap\fB$.
 \end{enumerate}
\end{lemma}
\begin{proof}\strut
 The first part is proved by induction on the derivation of
 $\GA\vdlo\alpha(M):\fA$.  For the second part, we use confluence of
 $\beta$-reduction.  \qed
\end{proof}

Now the logical approximation theorem can be formulated as follows.


\begin{theorem}
 $\GA\vdo M:\fA$ if and only if there exists $M'$ such that $M\bredstar
 M'$ and $\GA\vdo\alpha(M'):\fA$.
\end{theorem}
\begin{proof}
 ($\Rightarrow$) By Theorem~\ref{theorem:no-eq-so-eq-lo}, it suffices
 to show that if $\GA\vdlo M:\fA$ then there exists $M'$ such that
 $M\bredstar M'$ and $\GA\vdlo\alpha(M'):\fA$.  The proof is by
 induction on the derivation of $\GA\vdlo M:\fA$.
 Here we consider some cases.
 \begin{itemize}
  \renewcommand{\labelitemi}{$\bullet$}
  \item \raisebox{-1.5ex}[3ex][3ex]{$
	\infer[(\mathsf{Beta})^{\ell}]{\GA\vdlo(\lam{x}{M})NN_1\dots
	N_n:\fA}{
	 \GA\vdlo M\subst{x}{N}N_1\dots N_n:\fA}
	$} \\
	By the induction hypothesis, there exists $M'$ such that
	$M\subst{x}{N}N_1\dots N_n\bredstar M'$ and
	$\GA\vdlo\alpha(M'):\fA$.  This $M'$ also satisfies
	$(\lam{x}{M})NN_1\dots N_n\bredstar M'$.
  \item \raisebox{-1.5ex}[2.5ex][3ex]{$
	\infer[(\mathsf{L}\to)]{\GA,
	x:\fA_1\to\fA_2\vdlo xNN_1\dots N_n:\fB}{
	 \GA\vdlo N:\fA_1 &
	 \GA,y:\fA_2\vdlo yN_1\dots N_n:\fB}
	$} \\
	where $y\notin\FV{N_1}\cup\dots\cup\FV{N_n}$ and $y\notin\GA$.
	By the induction hypothesis, there exist $N',N_1',\dots,N_n'$
	such that $N\bredstar N'$, $N_i\bredstar N_i'$,
	$\GA\vdlo\alpha(N'):\fA_1$ and $\GA,y:\fA_2\vdlo
	y\alpha(N_1')\dots\alpha(N_n'):\fB$.  Hence, by an
	instance of the $(\mathsf{L}\to)$ rule, we obtain
	$\GA,x:\fA_1\to\fA_2\vdlo x\alpha(N')\alpha(N_1')\dots
	\alpha(N_n'):\fB$. So we take $xN'N_1'\dots N_n'$ as $M'$.
  \item \raisebox{-1.5ex}[3ex][3ex]{$
	\infer[(\mathsf{R}\to)]{\GA\vdlo\lam{x}{N}:\fA\to\fB}{
	 \GA,x:\fA\vdlo N:\fB}
	$} \\
	where $x\notin\GA$.  By the induction hypothesis, there exists
	$N'$ such that $N\bredstar N'$ and
	$\GA,x:\fA\vdlo\alpha(N'):\fB$.  By an instance of the
	$(\mathsf{R}\to)$ rule, we obtain
	$\GA\vdlo\lam{x}{\alpha(N')}:\fA\to\fB$.
	Since $\alpha(\lam{x}{N'})\equiv\lam{x}{\alpha(N')}$,
	we take $\lam{x}{N'}$ as $M'$.
  \item \raisebox{-1.5ex}[3ex][3ex]{$
	\infer[(\mathsf{R}\,\cap)]{\GA\vdlo M:\fA\cap\fB}{
	 \GA\vdlo M:\fA & \GA\vdlo M:\fB}
	$} \\
	By the induction hypothesis, there exist $M_1,M_2$
	such that $M\bredstar M_1$, $M\bredstar M_2$,
	$\GA\vdlo\alpha(M_1):\fA$ and $\GA\vdlo\alpha(M_2):\fB$.
	Then by Lemma~\ref{lemma:approx}(\ref{inters-app}), there exists
	$M'$ such that $M\bredstar M'$ and $\GA\vdlo\alpha(M'):\fA\cap\fB$.
 \end{itemize}
 ($\Leftarrow$) We can show by induction on the derivation that if
 $\GA\vdo\alpha(M'):\fA$ then $\GA\vdo M':\fA$.
 Hence, by Lemma~\ref{lemma:sub-exp}, we have $\GA\vdo M:\fA$.  \qed
\end{proof}

Thus our method has been successfully applied to proving the
approximation theorem for the mapping $\alpha$ and the system
$\LA_{\cap\omega}$.  It is work in progress to give similar proofs
of the approximation theorems for the $\eta$-approximation mapping
$\alpha_{\eta}$, which maps $\lam{x}{\bot}$ directly to $\bot$, and
type systems with various preorders as discussed in
\cite{MR2001e:03033,MR2003d:03016,Barendregt2013}.
\section{Conclusion}\label{section:concl}

We have presented uniform proofs of the characterisation theorems of
normalisation properties and the approximation theorem.  The proofs
have been given via intersection type systems in sequent calculus
style. 
As investigated in \cite{Kikuchi2009}, our method can be considered to
have embedded certain conditions for applying reducibility 
directly into the typing rules of the sequent calculus style systems.
(See \cite{KamareddineRW12} for a recent survey of general conditions
for applying the reducibility method.)

As mentioned in the introduction, there are some proofs
\cite{vRS1995,Matthes2000,David2001,Abel2007} of strong normalisation
for terms typable with intersection types without using reducibility,
but they have not considered any other properties than normalisation.
Other syntactic proofs of strong normalisation for terms typable with
intersection types are found in \cite{KfouryWells1995,Boudol2003},
where the problem is reduced to that of weak normalisation with
respect to another calculus or to another notion of reduction.  The
proofs of \cite{vRS1995,Valentini2001} and ours are different from
those of \cite{KfouryWells1995,Boudol2003} in that strong
normalisation is proved directly rather than inferring it from weak
normalisation.  Yet another syntactic proof
\cite{BucciarelliPipernoSalvo2003} uses a translation from terms
typable with intersection types into simply typed $\l$-terms.


There are many directions for future work.
In addition to
the one indicated at the last paragraph of Section~\ref{section:app},
it would be worth investigating
the type inference and the inhabitation problems
for intersection types
by means of our sequent calculus style systems.




\vspace{3ex}
\noindent
\textbf{Acknowledgements\ }
I would like to thank Katsumasa Ishii for drawing my attention to
Valentini's paper and pointing out that the system includes the
$\eta$-rule.  I also thank the anonymous reviewers of ITRS 2014
workshop for valuable comments.  The figures of the derivations have
been produced with Makoto Tatsuta's \texttt{proof.sty} macros.

\nocite{*}
\bibliographystyle{eptcs}
\bibliography{main}

\end{document}